\documentclass[amscd,amssymb,10pt]{amsart}

\usepackage{amsrefs,enumitem,graphics,graphicx,mathrsfs,mathtools,soul,xypic}
\usepackage[margin = 1in]{geometry}
\usepackage[usenames,dvipsnames]{color}
\usepackage[pdftex,breaklinks,colorlinks,citecolor = green,urlcolor = blue]{hyperref}

\allowdisplaybreaks
\setlist[itemize]{leftmargin = *}
\setlist[enumerate]{leftmargin = *}

\BibSpec{collection.article}{
    +{}  {\PrintAuthors}                {author}
    +{,} { \textit}                     {title}
    +{.} { }                            {part}
    +{:} { \textit}                     {subtitle}
    +{,} { \PrintContributions}         {contribution}
    +{,} { \PrintConference}            {conference}
    +{}  { \PrintBook}                  {book}
    +{,} { }                            {booktitle}
    +{,} { \PrintEditorsA}              {editor}
    +{,} { }                            {publisher}
    +{,} { }                            {address}
    +{,} { \PrintDateB}                 {date}
    +{,} { pp. }                        {pages}
    +{,} { }                            {status}
    +{,} { available at \eprint}        {eprint}
    +{}  { \parenthesize}               {language}
    +{}  { \PrintTranslation}           {translation}
    +{;} { \PrintReprint}               {reprint}
    +{.} { }                            {note}
    +{.} {\SentenceSpace \PrintReviews} {review}
}

\newtheorem*{Prop}{Proposition}

\newtheorem*{FlatoSimonThm}{The Flato-Simon Integrability Criterion}
\makeatletter \newenvironment{Flato|Simon}{\FlatoSimonThm \phantomsection \def \@currentlabel{Flato-Simon Integrability Criterion}}{\endFlatoSimonThm} \makeatother

\newtheorem*{NelsonThm}{Nelson's Analytic-Vectors Theorem}
\makeatletter \newenvironment{Nelson}{\NelsonThm \phantomsection \def \@currentlabel{Nelson's theorem}}{\endNelsonThm} \makeatother

\newtheorem*{StoneVonNeumannThm}{The Stone-Von Neumann Theorem}
\makeatletter \newenvironment{Stone|Von Neumann}{\StoneVonNeumannThm \phantomsection \def \@currentlabel{Stone-von Neumann Theorem}}{\endStoneVonNeumannThm} \makeatother

\theoremstyle{definition}
\newtheorem{Def}{Definition}
\newtheorem*{Eg}{Example}

\newenvironment{proof*}[1][\proofname]{ \begin{proof}[#1]}{\end{proof}}

\newcommand{\C}{\mathbb{C}}
\newcommand{\g}{\mathfrak{g}}
\newcommand{\N}[1]{\mathbb{N}_{#1}}
\newcommand{\R}[1]{\mathbb{R}^{#1}}
\newcommand{\df}{\stackrel{\textnormal{df}}{=}}
\newcommand{\Id}[1]{\textnormal{Id}_{#1}}
\newcommand{\h}{\mathfrak{h}_{3}}
\newcommand{\Bdd}[1]{\mathscr{B} \! \( #1 \)}
\newcommand{\doi}[1]{ doi: \href{http://dx.doi.org/\detokenize{#1}}{\detokenize{#1}}.}
\newcommand{\Dom}[1]{\textnormal{Dom} \! \( #1 \)}
\newcommand{\Ska}[1]{{\mathscr{L}_{\textnormal{ska}}} \! \( #1 \)}
\newcommand{\Sks}[1]{{\mathscr{L}_{\textnormal{sks}}} \! \( #1 \)}
\newcommand{\Uni}[1]{\mathscr{U} \! \( #1 \)}
\newcommand{\Norm}[1]{\left\| #1 \right\|}
\newcommand{\Inner}[2]{\left\langle #1,#2 \right\rangle}

\newcommand{\Mapping}[4]{\left\{ \begin{matrix} #1 & \to & #2 \\ #3 & \mapsto & #4 \end{matrix} \right\}}

\renewcommand{\(}{\left(}
\renewcommand{\)}{\right)}
\renewcommand{\[}{\left[}
\renewcommand{\]}{\right]}
\renewcommand{\d}{\mathrm{d}}
\renewcommand{\H}{\mathbf{H}_{3}}
\renewcommand{\L}{{L^{2}}(\mathbb{R})}
\renewcommand{\S}{\mathcal{S}(\mathbb{R})}
\renewcommand{\MR}[1]{~\href{http://www.ams.org/mathscinet-getitem?mr=#1}{MR #1}.}

\def\equationautorefname~#1\null{(#1)\null}

\title{An Infinitesimal Version of the Stone-Von Neumann Theorem}
\author{Leonard Huang}
\address[]{Department of Mathematics \\ University of Colorado Boulder \\ 2300 Colorado Ave \\ Boulder, Colorado 80309}
\email[]{leonard.huang@colorado.edu}


\begin{document}


\begin{abstract}
In this paper, we present an infinitesimal version of the Stone-von Neumann Theorem. This work was motivated by the need to formulate the uniqueness property of the Heisenberg Commutation Relation purely in terms of unbounded operators.
\end{abstract}


\maketitle


\section{Introduction}

Traditionally, the uniqueness property of the Heisenberg Commutation Relation (HCR) has often been identified with that of the Weyl Commutation Relation (WCR). There are several difficulties with this, mathematically and physically. To be mathematically precise, the correspondence between the two relations is only formal --- the WCR gives rise to the HCR but not the other way round, as the HCR is not integrable to the WCR in general. From the physical point of view, the HCR has a direct bearing on quantum mechanics, built into the theory itself and yielding important consequences such as the Heisenberg Uncertainty Principle, whereas the WCR appears devoid of any direct physical meaning \cite{Flato,Schmudgen}. Our aim in this paper, then, is to address these deficiencies by formulating the uniqueness property of the HCR solely in terms of unbounded operators. In doing so, we will employ powerful results in the theory of Lie-algebra representations.


\section{Prerequisites from Unbounded Operator Theory}

The following definitions and facts are given for an arbitrary Hilbert space $ \mathcal{H} $.
\begin{itemize}
\item
An \textbf{unbounded} operator on $ \mathcal{H} $ is a $ \C $-linear mapping from a dense linear subspace of $ \mathcal{H} $ to $ \mathcal{H} $.

\item
If $ A $ and $ B $ are unbounded operators on $ \mathcal{H} $, then we say that $ B $ is an \textbf{extension} of $ A $ and write $ A \leq B $ iff $ \Dom{A} \subseteq \Dom{B} $ and $ A v = B v $ for all $ v \in \Dom{A} $.

\item
If $ A $ is an unbounded operator on $ \mathcal{H} $ and $ D $ the set of all elements $ w \in \mathcal{H} $ for which the $ \C $-linear mapping
$$
\left\{ \begin{matrix} \Dom{A} & \to & \mathcal{H} \\ v & \mapsto & \Inner{A v}{w} \end{matrix} \right\}
$$
is bounded, then the \textbf{adjoint} of $ A $, denoted by $ A^{*} $, is defined as the $ \C $-linear mapping with domain $ D $ that takes each element $ w \in D $ to a unique element $ x \in \mathcal{H} $ (depending on $ w $) satisfying
$$
\forall v \in \Dom{A}: \quad
\Inner{A v}{w} = \Inner{v}{x}.
$$
The existence and uniqueness of $ x $ are guaranteed by the Riesz-Fr\'echet Theorem.

\item
Note that $ A^{*} $ may not be an unbounded operator on $ \mathcal{H} $ because $ \Dom{A^{*}} $ may not be dense in $ \mathcal{H} $.

\item
An unbounded operator $ A $ on $ \mathcal{H} $ is called \textbf{symmetric} iff
$$
\forall v,w \in \Dom{A}: \quad
\Inner{A v}{w} = \Inner{v}{A w}.
$$

\item
An unbounded operator $ A $ on $ \mathcal{H} $ is called \textbf{skew-symmetric} iff
$$
\forall v,w \in \Dom{A}: \quad
\Inner{A v}{w} = - \Inner{v}{A w}.
$$

\item
If $ A $ is a symmetric (resp. skew-symmetric) operator on $ \mathcal{H} $, then it is true that $ A \leq A^{*} $ (resp. $ - A \leq A^{*} $).

\item
A symmetric (resp. skew-symmetric) operator $ A $ on $ \mathcal{H} $ is called \textbf{self-adjoint} (resp. \textbf{skew-adjoint}) iff
$$
A = A^{*} \quad (\text{resp. $ - A = A^{*} $}).
$$

\item
A symmetric operator on $ \mathcal{H} $ that has a unique self-adjoint extension is called \textbf{essentially self-adjoint}.

\item
The set of all skew-symmetric (resp. skew-adjoint) operators on $ \mathcal{H} $ is denoted by $ \Sks{\mathcal{H}} $ (resp. $ \Ska{\mathcal{H}} $).

\item
If $ T $ and $ A $ are bounded and unbounded operators on $ \mathcal{H} $ respectively, we say that $ T $ \textbf{commutes} with $ A $ iff $ T A \leq A T $. (This implicitly means that $ T[\Dom{A}] \subseteq \Dom{A} $.)
\end{itemize}


\section{Lie-Algebra Representations}

All Lie algebras mentioned in this paper are real (i.e., their base field is $ \R{} $) and finite-dimensional.


\begin{Def}
If $ \mathcal{S} $ is a set of unbounded operators on a Hilbert space $ \mathcal{H} $, then a linear subspace $ D $ of $ \mathcal{H} $ is called an $ \mathcal{S} $-\textbf{invariant domain} iff $ D \subseteq \Dom{A} $ and $ A[D] \subseteq D $ for each $ A \in \mathcal{S} $.
\end{Def}


\begin{Def} \label{Lie-Algebra Representation}
If $ \g $ is a Lie algebra, then a triple $ (\rho,\mathcal{H},D) $ is called a \textbf{(Lie-algebra) representation} of $ \g $ iff:
\begin{enumerate}
\item[(1)]
$ \mathcal{H} $ is a Hilbert space.

\item[(2)]
$ \rho $ is a mapping from $ \g $ to $ \Sks{\mathcal{H}} $.

\item[(3)]
$ D $ is a dense $ \rho[\g] $-invariant domain.

\item[(4)]
$ \forall g_{1},g_{2} \in \g, ~ \forall \alpha \in \R{}: \quad \rho(\alpha g_{1} + g_{2})|_{D} = \alpha \cdot \rho(g_{1})|_{D} + \rho(g_{2})|_{D} $.

\item[(5)]
$ \forall g_{1},g_{2} \in \g: \quad \rho([g_{1},g_{2}])|_{D} = [\rho(g_{1})|_{D},\rho(g_{2})|_{D}] $.
\end{enumerate}
\end{Def}


Lie-algebra representations are closely related to many areas of physics, especially quantum mechanics, as the next example shows.


\begin{Eg}[The Heisenberg Lie Algebra]
Let $ \h $ denote the three-dimensional Heisenberg Lie Algebra. This is the real Lie algebra spanned by three elements $ x $, $ y $ and $ z $ that satisfy the following Lie-bracket relations:
$$
[x,y] = z, \quad
[x,z] = 0  \quad \text{and} \quad
[y,z] = 0.
$$
Let $ Q $ and $ P_{\hbar} $ denote, respectively, the (self-adjoint) quantum-mechanical position and momentum operators on $ \L $. The subscript ``$ \hbar $'' in ``$ P_{\hbar} $'' plays the role of a parameter taking on small positive real values, and it indicates that we are employing the physicist's definition of the momentum operator as the differential operator $ \dfrac{\hbar}{i} \cdot \dfrac{\d}{\d x} $, whose domain is the Sobolev space $ {W^{1,2}}(\R{}) $. We call $ (Q,P_{\hbar}) $ the \textbf{Schr\"odinger pair}.

Let $ \S $ denote the space of Schwartz functions on $ \R{} $. Then $ \( \sigma,\L,\S \) $ is a representation of $ \h $, where $ \sigma: \h \to \Sks{\L} $ is defined by
$$
\forall (\alpha,\beta,\gamma) \in \R{3}: \quad
\sigma(\alpha x + \beta y + \gamma z) \df - i \! \( \alpha Q|_{\S} + \beta P_{\hbar}|_{\S} + \gamma \hbar \cdot \Id{\S} \).
$$

Verifying Conditions (1)-(4) in \autoref{Lie-Algebra Representation} is a rather routine exercise. It is Condition (5) that deserves attention --- it holds because the Schr\"odinger pair obeys the following commutation relation on $ \S $:
$$
\forall \varphi \in \S: \quad
[Q,P_{\hbar}] \varphi \df (Q P_{\hbar} - P_{\hbar} Q) \varphi = i \hbar \cdot \varphi.
$$
This is the famous \textbf{Heisenberg Commutation Relation (HCR)} in quantum mechanics. It can be placed in the context of a general Hilbert space, as shown in the next definition.
\end{Eg}


\begin{Def}
A pair $ (A,B) $ of self-adjoint operators on a Hilbert space $ \mathcal{H} $ is said to \textbf{obey the HCR} on a dense $ \{ A,B \} $-invariant domain $ D $ iff
$$
\forall v \in D: \quad
[A,B] v \df (A B - B A) v = i \hbar \cdot v.
$$
\end{Def}


\section{Analytic Vectors and Integrability}


\begin{Def}
If $ A $ is an unbounded operator on a Hilbert space $ \mathcal{H} $, then a vector $ v \in \mathcal{H} $ is called $ A $-\textbf{analytic} iff:
\begin{itemize}
\item
$ v \in \Dom{A^{k}} $ for all $ k \in \N{} $.

\item
There exists an $ s > 0 $ such that $ \displaystyle \sum_{k = 0}^{\infty} \frac{s^{k}}{k!} \Norm{A^{k} v} < \infty $, where $ {A^{0}}(v) \df v $.
\end{itemize}
If $ \mathcal{S} $ is a set of unbounded operators on $ \mathcal{H} $, then a vector $ v \in \mathcal{H} $ is called $ \mathcal{S} $-\textbf{analytic} iff $ v $ is $ A $-analytic for each $ A \in \mathcal{S} $.
\end{Def}


Analytic vectors constitute an important concept because they provide an exact criterion for deciding whether or not a symmetric operator is essentially self-adjoint. This criterion, which is given next, was first announced by Edward Nelson in \cite{Nelson}.


\begin{Nelson} \label{Nelson}
Let $ A $ be a symmetric operator on a Hilbert space $ \mathcal{H} $. Then $ A $ is essentially self-adjoint iff $ \Dom{A} $ contains a dense set of $ A $-analytic vectors.
\end{Nelson}


Let $ \g $ be a Lie algebra. By the Lie-Cartan Theorem (also known as Lie's Third Theorem), there exists a simply connected Lie group, unique up to isomorphism, whose Lie algebra of left-invariant vector fields is isomorphic to $ \g $. We denote this Lie group by $ G_{\g} $.

The three-dimensional simply connected Lie group associated with $ \h $ is called the \textbf{(three-dimensional) Heisenberg Lie Group}. It is denoted by $ \H $ instead of $ G_{\h} $.


\begin{Def}
A \textbf{one-parameter unitary group} on a Hilbert space $ \mathcal{H} $ is a unitary representation $ U $ of $ \R{} $ on $ \mathcal{H} $, i.e., $ U $ is a homomorphism from $ \R{} $ to $ \Uni{\mathcal{H}} $. We call $ U $ \textbf{strongly continuous} iff it is strongly continuous as a unitary group representation.
\end{Def}


\begin{Def}
If $ U $ is a strongly continuous one-parameter unitary group on a Hilbert space $ \mathcal{H} $, then
$$
D_{U} \df \left\{ v \in \mathcal{H} ~ \middle| ~ \text{The limit $ \lim_{t \to 0} \[ \frac{U(t) - \Id{\mathcal{H}}}{t} \] v $ exists} \right\}
$$
is a dense linear subspace of $ \mathcal{H} $. We can thus define an unbounded operator $ A $ on $ \mathcal{H} $ with domain $ D_{U} $ by
$$
\forall v \in D_{U}: \quad
A v \df \lim_{t \to 0} \[ i \cdot \frac{U(t) - \Id{\mathcal{H}}}{t} \] v.
$$
Stone's Theorem (\cite{Riesz|Sz.-Nagy}, Section 137) guarantees that $ A $ is skew-adjoint. We then call $ A $ the \textbf{infinitesimal skew-adjoint generator} of $ U $.
\end{Def}


\begin{Def}
A representation $ (\rho,\mathcal{H},D) $ of a Lie algebra $ \g $ is called \textbf{integrable} iff there exists a strongly continuous unitary representation $ \pi $ of $ G_{\g} $ on $ \mathcal{H} $ such that $ {\partial \pi}(g)|_{D} = \rho(g)|_{D} $ for each $ g \in \g $, in which case we call $ \pi $ an \textbf{integral} of $ (\rho,\mathcal{H},D) $. Here, $ \partial \pi $ denotes the differential of $ \pi $ --- the mapping from $ \g $ to $ \Ska{\mathcal{H}} $ that takes each element $ g \in \g $ to the infinitesimal skew-adjoint generator of the strongly continuous one-parameter unitary group
$$
\Mapping{\R{}}{\Uni{\mathcal{H}}}{t}{\pi(\exp(t g))}.
$$
\end{Def}


There exist Lie-algebra representations that are not integrable; a famous example is given in \cite{Nelson}.

The following theorem --- a slight variant of a result due to Mosh\'e Flato and Jacques Simon --- gives a sufficient condition for the integrability of a Lie-algebra representation in terms of analytic vectors, which is yet another piece of evidence for their importance.


\begin{Flato|Simon} \label{Flato|Simon} \cite{Flato|Simon,Gotay}
Let $ (\rho,\mathcal{H},D) $ be a representation of a Lie algebra $ \g $. Let $ \mathcal{G} $ be a set of Lie generators for $ \g $. If $ D $ contains a dense set of $ \rho[\mathcal{G}] $-analytic vectors, then $ (\rho,\mathcal{H},D) $ is integrable and has a unique integral.
\end{Flato|Simon}


\section{An Infinitesimal Version of the Stone-Von Neumann Theorem}

In this section, we present an infinitesimal version of the Stone-von Neumann Theorem. Before we do so, let us give a few necessary definitions and also a statement of the Stone-von Neumann Theorem.


\begin{Def}
A pair $ (U,V) $ of strongly continuous one-parameter unitary groups on a Hilbert space $ \mathcal{H} $ is called \textbf{jointly irreducible} iff the union $ \{ U(s) \mid s \in \R{} \} \cup \{ V(t) \mid t \in \R{} \} $ is an irreducible set of unitary operators on $ \mathcal{H} $.
\end{Def}


\begin{Def}
Let $ (A_{1},B_{1}) $ and $ (A_{2},B_{2}) $ be pairs of unbounded operators on Hilbert spaces $ \mathcal{H}_{1} $ and $ \mathcal{H}_{2} $ respectively. We call $ (A_{1},B_{1}) $ and $ (A_{2},B_{2}) $ \textbf{unitarily equivalent} iff there is a unitary operator $ W: \mathcal{H}_{1} \to \mathcal{H}_{2} $ such that $ W A_{1} W^{-1} = A_{2} $ and $ W B_{1} W^{-1} = B_{2} $. This means that $ W $ is a bijection from $ \Dom{A_{1}} $ to $ \Dom{A_{2}} $ and from $ \Dom{B_{1}} $ to $ \Dom{B_{2}} $.
\end{Def}


\begin{Stone|Von Neumann} \label{Stone|Von Neumann}
Let $ (U,V) $ be a pair of strongly continuous one-parameter unitary groups on a separable Hilbert space $ \mathcal{H} $. Suppose that $ (U,V) $ is jointly irreducible and that
\begin{equation} \label{Weyl Commutation Relation}
\forall s,t \in \R{}: \quad
U(s) V(t) = e^{i \hbar s t} \cdot V(t) U(s).
\end{equation}
If $ A $ and $ B $ denote the infinitesimal skew-adjoint generators of $ V $ and $ U $ respectively (please note the order!), then $ (i A,i B) $ and $ (Q,P_{\hbar}) $ are unitarily equivalent.
\end{Stone|Von Neumann}


\begin{proof}
For a basic proof, see Lemma 2.1 and Theorem 3.1 of \cite{Takhtajan}. There are many other proofs, all of which vary in terms of the level of abstraction. A proof via Mackey's Imprimitivity Theorem can be found in \cite{Taylor}, while one based on the theory of Morita-Rieffel equivalence can be found in \cite{Raeburn|Williams}.
\end{proof}


\autoref{Weyl Commutation Relation} is known as the \textbf{Weyl Commutation Relation (WCR)}. The \ref{Stone|Von Neumann} thus says: Given a pair $ (U,V) $ of strongly continuous one-parameter unitary groups on a Hilbert space $ \mathcal{H} $ that obeys the WCR, if $ (A,B) $ denotes the corresponding pair of infinitesimal skew-adjoint generators, then $ (i A,i B) $ is a pair of self-adjoint operators on $ \mathcal{H} $ that obeys the HCR on a dense $ \{ A,B \} $-invariant domain $ D $.

In general, the correspondence is irreversible: If $ (A,B) $ is a pair of self-adjoint operators on $ \mathcal{H} $ that obeys the HCR on a dense $ \{ A,B \} $-invariant domain $ D $, then \emph{unless $ D $ contains a dense set of $ \{ A,B \} $-analytic vectors}, $ (A,B) $ may not necessarily correspond, in the manner described above, to a pair of strongly continuous one-parameter unitary groups on $ \mathcal{H} $ that obeys the WCR \cite{Schmudgen}. In cases where the correspondence can be reversed, we say that \textbf{the HCR is integrable to the WCR}.

The proof by John von Neumann and Marshall Stone of the uniqueness property of the HCR in \cite{Von Neumann,Stone} requires the assumption that the HCR is integrable to the WCR, which guarantees a two-way correspondence between the two relations. For lack of tools, they were unable to formulate this integrability assumption in infinitesimal terms (i.e., in terms of unbounded operators), and the status-quo remained unchanged until the appearance of results by Franz Rellich and Jacques Dixmier \cite{Rellich,Dixmier}, which allowed an almost-complete formulation of the uniqueness property of the HCR without making any reference at all to the WCR.

We are now ready for our main proposition, which does not seem to have appeared in the literature yet.


\begin{Prop}
Let $ A $ and $ B $ be self-adjoint operators on a separable Hilbert space $ \mathcal{H} $. Then $ (A,B) $ and $ (Q,P_{\hbar}) $ are unitarily equivalent iff the following conditions hold:
\begin{enumerate}[font = \normalfont]
\item[(1)]
There exists a dense $ \{ A,B \} $-invariant domain $ D $ such that
\begin{itemize}
\item
$ (A,B) $ obeys the HCR on $ D $ and

\item
$ D $ contains a dense set of $ \{ A,B \} $-analytic vectors.
\end{itemize}

\item[(2)]
If $ T \in \Bdd{\mathcal{H}} $ commutes with both $ A $ and $ B $, then $ T \in \C \cdot \Id{\mathcal{H}} $.
\end{enumerate}
\end{Prop}


\begin{proof}
Suppose that $ A $ and $ B $ satisfy Conditions (1) and (2).

Define a representation $ (\rho,\mathcal{H},D) $ of $ \h $ by
$$
\forall (\alpha,\beta,\gamma) \in \R{3}: \quad
\rho(\alpha x + \beta y + \gamma z) \df - i (\alpha A|_{D} + \beta B|_{D} + \gamma \hbar \cdot \Id{D}).
$$
This is indeed a representation because $ (A,B) $ obeys the HCR on $ D $.

As $ \{ x,y \} $ is a set of Lie generators for $ \h $, and as $ D $ contains a dense set of $ \{ A,B \} $-analytic vectors, the \ref{Flato|Simon} says that $ (\rho,\mathcal{H},D) $ is integrable and has a unique integral $ \pi $. Then by the definition of integrability, we have
$$
{\partial \pi}(x)|_{D} = - i A|_{D}, \quad
{\partial \pi}(y)|_{D} = - i B|_{D}  \quad \text{and} \quad
{\partial \pi}(z)|_{D} = - i \hbar \cdot \Id{D}.
$$
According to \ref{Nelson}, $ A|_{D} $, $ B|_{D} $ and $ \Id{D} $ are essentially self-adjoint,\footnote{Nelson's theorem is not really needed to show that $ \Id{D} $ is essentially self-adjoint. Simply observe that the closure of the graph of $ \Id{D} $ in $ \mathcal{H} \times \mathcal{H} $ is the graph of $ \Id{\mathcal{H}} $, which is obviously self-adjoint. A theorem by von Neumann then says that if the closure of the graph of a symmetric operator $ A $ is the graph of a self-adjoint operator $ B $, then $ A $ is automatically essentially self-adjoint and its unique self-adjoint extension is $ B $.} so
\begin{equation} \label{Generators}
{\partial \pi}(x) = - i A, \quad {\partial \pi}(y) = - i B \quad \text{and} \quad {\partial \pi}(z) = - i \hbar \cdot \Id{\mathcal{H}}.
\end{equation}
Before proceeding further, let us establish a claim. \\

\noindent \textbf{Claim:} $ \pi $ is an irreducible unitary representation of $ \H $.

\begin{proof*}[Proof of Claim]
Define strongly continuous one-parameter unitary groups $ U $ and $ V $ on $ \mathcal{H} $ by
$$
U \df \Mapping{\R{}}{\Uni{\mathcal{H}}}{t}{\pi(\exp(t y))} \quad \text{and} \quad
V \df \Mapping{\R{}}{\Uni{\mathcal{H}}}{t}{\pi(\exp(t x))}.
$$
From \autoref{Generators}, we see that $ U $ and $ V $ have infinitesimal skew-adjoint generators $ - i B $ and $ - i A $ respectively.

If $ T \in \Bdd{\mathcal{H}} $ commutes with $ \pi(h) $ for all $ h \in \H $, then $ T $ commutes with $ U(t) $ and $ V(t) $ for all $ t \in \R{} $. By Stone's Theorem (\cite{Riesz|Sz.-Nagy}, Section 137), $ T $ commutes with $ A $ and $ B $, so $ T \in \C \cdot \Id{\mathcal{H}} $, thanks to Condition (2). Therefore, $ \pi $ is an irreducible unitary representation of $ \H $ by Schur's Lemma (\cite{Taylor}, Proposition B.2).
\end{proof*}

By Stone's Theorem again, we have $ \pi(\exp(t z)) = e^{- i \hbar t} \cdot \Id{\mathcal{H}} $ for all $ t \in \R{} $. The BCH Formula then yields
\begin{align*}
\forall (\alpha,\beta) \in \R{2}: \qquad \quad
  ~ & \pi(\exp(\alpha x)) ~ \pi(\exp(\beta y)) \\
= ~ & \pi(\exp(\alpha x) \exp(\beta y)) \\
= ~ & \pi \! \( \exp \! \( \alpha x + \beta y + \frac{1}{2} [\alpha x,\beta y] \) \) \\
= ~ & \pi \! \( \exp \! \( \alpha x + \beta y + \frac{1}{2} \alpha \beta z \) \) \quad \quad (\text{As $ [x,y] = z $.}) \\
= ~ & \pi \! \( \exp(\alpha \beta z) \exp \! \( \alpha x + \beta y - \frac{1}{2} \alpha \beta z \) \) \qquad (\text{As $ z $ is a central element of $ \h $.}) \\
= ~ & \pi \! \( \exp(\alpha \beta z) \exp \! \( \beta y + \alpha x + \frac{1}{2} [\beta y,\alpha x] \) \) \qquad (\text{As $ [y,x] = - z $.}) \\
= ~ & \pi(\exp(\alpha \beta z) \exp(\beta y) \exp(\alpha x)) \\
= ~ & \pi(\exp(\alpha \beta z)) ~ \pi(\exp(\beta y)) ~ \pi(\exp(\alpha x)) \\
= ~ & e^{- i \hbar \alpha \beta} \cdot \pi(\exp(\beta y)) ~ \pi(\exp(\alpha x)).
\end{align*}
It follows from this and the irreducibility of $ \pi $ that $ (U,V) $ is jointly irreducible and obeys the WCR:
$$
\forall (\alpha,\beta) \in \R{2}: \quad
\pi(\exp(\beta y)) ~ \pi(\exp(\alpha x)) = e^{i \hbar \alpha \beta} \cdot \pi(\exp(\alpha x)) ~ \pi(\exp(\beta y)).
$$
Therefore, $ (A,B) $ and $ (Q,P_{\hbar}) $ are unitarily equivalent by the \ref{Stone|Von Neumann}.

For the other direction, suppose that $ (A,B) $ and $ (Q,P_{\hbar}) $ are unitarily equivalent.

By Proposition 2.1 of \cite{Takhtajan}, if $ T \in \Bdd{\L} $ commutes with both $ Q $ and $ P_{\hbar} $, then $ T \in \C \cdot \Id{\L} $. Hence, Condition (2) is satisfied.

Next, observe that $ \S $ is a dense $ \{ Q,P_{\hbar} \} $-invariant domain that contains the Hermite functions, whose linear span is a dense set of $ \{ Q,P_{\hbar} \} $-analytic vectors (\cite{Reed|Simon}, Section X.6, Example 2). Hence, under the unitary equivalence, $ \S $ transfers to a dense $ \{ A,B \} $-invariant domain $ D $ on which $ (A,B) $ obeys the HCR (because the Schr\"odinger pair obeys the HCR on $ \S $), and likewise, the linear span of the Hermite functions transfers to a dense subset of $ D $ consisting of $ \{ A,B \} $-analytic vectors.

The proof of the proposition is now complete.
\end{proof}


\section{Concluding Remarks}

It may be claimed that the infinitesimal version of the Stone-von Neumann Theorem that we have put forth suffers from the drawback that its proof still relies on the WCR and the Stone-von Neumann Theorem. This, however, appears unavoidable. After all, our goal was to formulate, not prove, the uniqueness property of the HCR solely in terms of unbounded operators. What we have done is merely to separate the physical and mathematical content of the HCR.

The proposition can be generalized to higher dimensions. Our proof easily goes through in these cases, for we only have to replace $ \h $ and $ \H $ by their higher-dimensional analogs and also use the higher-dimensional version of the Stone-von Neumann Theorem.




\begin{bibdiv}
\begin{biblist}

\bib*{Physics}{book}{
conference = {
             title   = {Proc. Internat. Sympos.},
             address = {Warsaw},
             date    = {1974},
             },
title      = {Mathematical Physics and Physical Mathematics},
editor     = {K. Maurin},
editor     = {R. R\c aczka},
date       = {1976},
publisher  = {Reidel},
address    = {Hingham, MA}
}

\bib*{Mechanics}{book}{
title     = {Mechanics: From Theory to Computation. Essays in Honor of Juan-Carlos Simo},
date      = {2000},
publisher = {Springer},
address   = {New York, NY}
}

\bib{Dixmier}{article}{
author   = {Dixmier, J.},
title    = {Sur la relation $ i (P Q - Q P) = 1 $},
journal  = {Compos. Math.},
volume   = {13},
date     = {1958},
pages    = {263--269},
review   = {\MR{0101478}}
}

\bib{Flato}{collection.article}{
author     = {Flato, M.},
title      = {Theory of Analytic Vectors and Applications},
xref       = {Physics},
date       = {1976},
pages      = {231--250},
review     = {\MR{0550918}}
}

\bib{Flato|Simon}{article}{
author   = {Flato, M.},
author   = {Simon, J.},
title    = {Separate and Joint Analyticity for Lie Groups Representations},
journal  = {J. Funct. Anal.},
volume   = {13},
date     = {1972},
pages    = {268--276},
review   = {\doi{10.1016/0022-1236(73)90035-9}}
}

\bib{Gotay}{collection.article}{
author    = {Gotay, M.},
title     = {Obstructions to Quantization},
xref      = {Mechanics},
date      = {2000},
pages     = {171--216},
review    = {\MR{1766362}}
}

\bib{Nelson}{article}{
author   = {Nelson, E.},
title    = {Analytic Vectors},
journal  = {Ann. of Math. (2)},
volume   = {70},
number   = {3},
date     = {1959},
pages    = {572--615},
review   = {\MR{107176}}
}

\bib{Von Neumann}{article}{
author   = {von Neumann, J.},
title    = {Die Eindeutigkeit der Schr\"odingerschen Operatoren},
journal  = {Math. Ann.},
volume   = {104},
date     = {1931},
pages    = {570--578},
review   = {\MR{1512685}}
}

\bib{Raeburn|Williams}{book}{
author    = {Raeburn, I.},
author    = {Williams, D.},
title     = {Morita Equivalence and Continuous-Trace $ C^{*} $-Algebras},
date      = {1991},
series    = {Math. Surveys Monogr.},
volume    = {60},
publisher = {Amer. Math. Soc.},
address   = {Providence, RI},
review    = {\MR{2288954}}
}

\bib{Reed|Simon}{book}{
author    = {Reed, M.},
author    = {Simon, B.},
title     = {Methods of Modern Mathematical Physics. II: Fourier Analysis, Self-Adjointness},
date      = {1975},
publisher = {Academic Press},
address   = {San Diego, CA},
review    = {\MR{0493420}}
}

\bib{Rellich}{article}{
author   = {Rellich, F.},
title    = {Der Eindeutigkeitssatz f\"ur die Los\"ungen der quantenmechanischen Vertauschungsrelationen},
journal  = {Nachr. Akad. Wiss. G\"ottingen Math.-Phys. Kl. II},
date     = {1946},
pages    = {107--115},
review   = {\MR{0022310}}
}

\bib{Riesz|Sz.-Nagy}{book}{
author    = {Riesz, F.},
author    = {Sz.-Nagy, B.}
title     = {Functional Analysis},
date      = {1990},
publisher = {Dover Publications},
address   = {Mineola, NY},
review    = {\MR{1068530}}
}

\bib{Schmudgen}{article}{
author   = {Schm\"udgen, K.},
title    = {On the Heisenberg Commutation Relation. I},
journal  = {J. Funct. Anal.},
volume   = {5},
date     = {1983},
pages    = {8--49},
review   = {\MR{0689997}}
}

\bib{Stone}{article}{
author   = {Stone, M.},
title    = {Linear Transformations in Hilbert Space. III. Operational Methods and Group Theory},
journal  = {Proc. Natl. Acad. Sci. USA},
volume   = {16},
date     = {1930},
pages    = {172--175},
review   = {\doi{10.1073/pnas.16.2.172}}
}

\bib{Takhtajan}{book}{
author    = {Takhtajan, L.},
title     = {Quantum Mechanics for Mathematicians},
date      = {2008},
series    = {Grad. Stud. Math.},
volume    = {95},
publisher = {Amer. Math. Soc.},
address   = {Providence, RI},
review    = {\MR{2433906}}
}

\bib{Taylor}{book}{
author    = {Taylor, M.},
title     = {Noncommutative Harmonic Analysis},
date      = {1986},
series    = {Math. Surveys Monogr.},
volume    = {22},
publisher = {Amer. Math. Soc.},
address   = {Providence, RI},
review    = {\MR{0852988}}
}

\end{biblist}
\end{bibdiv}
\end{document}